\documentclass{ifacconf}
\usepackage{amsmath}
\usepackage{amssymb}
\usepackage{mathabx}
\usepackage{tikz}
\usepackage[bb=boondox]{mathalfa}
\usepackage{comment}
\usepackage{xspace}
\usepackage{subfigure}
\usepackage{graphicx}      
\usepackage{natbib}        
\usetikzlibrary{arrows,automata,positioning,fit,calc,shapes}
\usepackage{amsfonts,amssymb, bm, bbm, subfigure}

\newcommand{\ba}{\begin{array}}
\newcommand{\ea}{\end{array}}

\newcommand{\be}{\begin{equation}}
\newcommand{\ee}{\end{equation}}

\newcommand{\mc}{\mathcal}

\newcommand{\1}{\mathbbm{1}}
\newcommand{\0}{\mathbb{0}}

\newcommand{\R}{\mathbb{R}}

\renewcommand{\P}{\mathbb{P}}
\def\x{\mathbf{x}}
\def\y{\mathbf{y}}

\DeclareMathOperator*{\argmax}{argmax}
\DeclareMathOperator*{\argmin}{argmin}

\def\qed{\hfill \vrule height 7pt width 7pt depth 0pt\medskip}
\def\beq{\begin{equation}}
\def\eeq{\end{equation}}

\newtheorem{theorem}{Theorem}
\newtheorem{definition}[theorem]{Definition}
\newtheorem{proposition}[theorem]{Proposition}
\newtheorem{lemma}[theorem]{Lemma}
\newtheorem{corollary}[theorem]{Corollary}
\theoremstyle{remark}
\newtheorem{remark}[theorem]{Remark}

\begin{document}
\begin{frontmatter}

\title{Controlling network coordination games}

\thanks[footnoteinfo]{Giacomo Como is also with the Department of Automatic Control, Lund University, Sweden. This work was partially supported by MIUR grant Dipartimenti di Eccellenza 2018--2022 [CUP: E11G18000350001], the Swedish Research Council, and by the Compagnia di San Paolo.}

\author[Second]{St\'ephane Durand} 
\author[Second]{Giacomo Como}
\author[Second]{Fabio Fagnani}

\address[Second]{Department of Mathematical Sciences ``G.L.~Lagrange'', Politecnico di Torino, Corso Duca degli Abruzzi 24, 10129 Torino, Italy\\ 
(e-mail: \{stephane.durand, giacomo.como, fabio.fagnani\}@polito.it).}

\begin{abstract} We study a novel control problem in the context of network coordination games: the individuation of the smallest set of players capable of driving the system, globally, from one Nash equilibrium to another one. Our main contribution is the design of a randomized algorithm based on a time-reversible Markov chain with provable convergence garantees.
 
\end{abstract}

\begin{keyword} Coordination games, Contagion, Control set, Resilience, Randomized search
\end{keyword}

\end{frontmatter}

\section{Introduction}
In a binary ($0-1$) coordination game over a graph, where initially all agents are in the Nash equilibrium $0$, what is the minimum number of agents that if forced to $1$ will push the system to converge to the Nash equilibrium of all $1$'s under a best response dynamics? This paper is devoted to the analysis of this problem and to the design of an algorithm for an efficient solution. 

The problem considered can be framed in the more general setting of studying minimal interventions strategies needed to make a multi-agent system governed by agents' myopic utility maximization, to drive from a Nash equilibrium to a desired another one. In game theory, typically, interventions have been modeled as perturbations of the utility functions (e.g. taxes and prices in economic models or tolls in transportation systems). Here we instead take a different viewpoint: that of individuating a subset of nodes (hopefully small) that if suitably controlled will lead the entire system to the desired configuration. The minimum cardinality of this set can also be interpreted as a measure of resilience of the system: the larger it is, the more difficult is for an external shock to destabilize it. 

The problem of determining the best set of nodes to exert the most effective control in a networked system has recently appeared in other contexts: for instance in \cite{Yildiz-2013} and \cite{Vassio-2014} authors study the problem of the optimal position of stubborn influencers in linear opinion dynamics.

Binary coordination games have received a great attention in the recent years as one of the basic models for games with strategic complementarities \cite{gamesonnetworks}. Its variegate applications include modeling of social and economic behaviors like adopt a new technology, participate in an event, provide a public good effort.

This game is analyzed in detail in \cite{morris} where the key concept of cohesiveness of a set of players is introduced and then used in characterizing all NE's. 
Moreover, the question if an initial seed of influenced players (that maintain action  $1$ in all circumstances) is capable of propagating to the all network is addressed in the same paper and an equivalent characterization of  this spreading phenomenon is also expressed in terms of cohesiveness. This contagion phenomenon is exactly what we want to analyze: subset of nodes from which propagation is successful will be called sufficient control sets and our goal is to find such sets of minimum possible cardinality. 

The condition proposed in \cite{morris} is computationally quite demanding and can not be used to directly solve our optimization problem. Indeed, even to determine if a single set  is a sufficient control set,  it requires a number of check growing exponentially in the cardinality of the complement of such set. 

The complementary problem of understanding, given an integer $k$, what is the maximum possible spreading of the state $1$, starting from an initial seed of $k$ influenced players, was studied in a seminal paper by \cite{Kempe}. While their problem and ours are related, they are independent, in the sense that solving one does not provide a solution of the other. Another point worth stressing is that, in their setting, \cite{Kempe} consider agents equipped with random independent activation thresholds and take as functional to be optimized the average size of the maximum spreading. They prove that such functional is sub-modular and then they design a greedy algorithm for obtaining suboptimal solutions. The randomness that they introduce is actually crucial in their approach, as the functional considered would not be sub-modular for deterministic choices of thresholds. 

In this paper we consider a scenario when all agents have a fixed threshold $1/2$, thus not covered in \cite{Kempe}, and we design an iterative search randomized algorithm with provable properties of convergence towards sufficient control sets of minimum cardinality. The core of the algorithm is a time-reversible Markov chains over the family of all sufficient control sets that starts with the full set, moves through all of them in an ergodic way, and concentrates its mass on those of minimum cardinality .

We conclude this introduction with a brief outline of the paper. In the final part of this section we report some basic notation used throughout the paper. Section \ref{sec:model} is dedicated to the formal introduction of the problem. The main technical parts are Sections \ref{sec:monotone} and \ref{sec:MC}. In Section \ref{sec:monotone}, we introduce the important notion of monotone crusade (appeared for other purposes in \cite{Drakopoulos.ea:2015,Drakopoulos.ea:2016}) and we give an equivalent (but more operative) characterization of sufficient control sets. In Section  \ref{sec:MC} we introduce a family of reversible Markov chains whose invariant probability is proven to concentrate on the optimal sufficient control sets. Section \ref{sec:simulations} describes the algorithm, based on the Markov chains introduced in the previous section, and presents some simulation results. Finally, Section \ref{sec:conclusions} ends the paper.

\subsection{Notation} Vectors are indicated in bold-face letters ${\bf x},\, {\bf y}, {\bf z}$. For $\mathbf{x},\, \mathbf{y}$ are two vectors of the same dimension, the notation $\mathbf{x} \le \mathbf{y}$ indicates that $\mathbf{x}$ is lower or equal \emph{component-wise} than $\mathbf{y}$.
We define as usual the binary vectors $\delta_i$: $(\delta_i)_i=1$ and $(\delta_i)_j=0$ for every $j\neq i$.
If $\mc S\subseteq\{1,\dots , n\}$, we put $\1_{\mc S}=\sum_{i\in\mc S}\delta_i$. 
Every ${\bf x}\in \{0,1\}^n$ can be written as ${\bf x}=\1_S$ for some $S\subseteq \{1,\dots , n\}$.  We call such a subset $S$ the \emph{support} of ${\bf x}$ and we denote it $S_{\bf x}$. 
We use the notation $\0$ and $\1$ to denote, respectively, the vector of all $0$'s and the vector of all $1$'s of any possible dimension.

{
\section{Controlled majority dynamics}\label{sec:model}
We consider a set of players $\mc V=\{1,\dots , n\}$ connected by a simple undirected graph $\mc G=(\mc V, \mc E)$ (e.g. $\mc E\subseteq\mc V\times\mc V$ is such that $(i,i)\not\in\mc E$ for any $i$ and $(i,j)\in\mc E$ iff $(j,i)\in\mc E$). The binary set $\mc A=\{0,1\}$ is the set of possible actions for all players. We put $\mc X=\mc A^n$ and we define
the \emph{majority game} on $\mc G$ as the game where each player $i\in \mc V$ has utility $\lambda^{\text c}_i:\mc X\to\R$ given by
%
%
$$\lambda^{\text c}_i({\bf x})=|\{j\in N_i\,|\, x_j=x_i\}|$$
where $N_i$ is the neighborhood of agent $i$ in $\mc G$. In other words, $\lambda^{\text c}_i({\bf x})$ is 
the number of neighbors of $i$ with which $i$ is in agreement.

As usual in game theory, given a configuration vector  $\x\in\mc X$ and a player $i$, we indicate with $\x_{-i}$ the configuration restricted to all players but $i$ and we consequently write 
$\x=(x_i, \x_{-i})$. 

Best response sets are defined by 
$$
\mc B^{\text c}_i(\bf x)=\argmax\limits_{\alpha\in\mc A} \lambda^{\text c}_i(\alpha, \x_{-i})
$$

Using the notation
\beq\label{n01}n_{i,a}(\mathbf{x})=|\{j\in N_i\,|\, x_j=\alpha\}|,\quad \alpha=0,1\eeq
to indicate the number of neighbors of an agent  $i$ playing action $\alpha$  in the configuration $\mathbf{x}$,
%
%
%
Best response sets can be more explicitly described as
$$
\mc B^{\text c}_i(\bf x)=\left\{\begin{array}{ll} \{0\},\quad&{\rm if}\, n_{i,0}(\mathbf{x})>n_{i,1}(\mathbf{x})\\
\{0,1\},\quad&{\rm if}\, n_{i,0}(\mathbf{x})=n_{i,1}(\mathbf{x})\\
\{1\},\quad&{\rm if}\, n_{i,0}(\mathbf{x})<n_{i,1}(\mathbf{x})
\end{array}\right.$$

%
%

$\mc N$ denotes the set of Nash equilibria: $\mc N=\{{\bf x}\in\mc X\,|\, x_i\in \mc B^{\text c}_i(\x)\,\forall i\in\mc V\}$. This set depends on the topology of the graph $\mc G$, note however that $\0$ and $\1$ are always Nash equilibria.

%
%

%
%
%
It is well known that this game is potential with a potential $\Phi_{\text c}$ given by
\beq \Phi_{\text c}({\bf x})=|\{(i,j)\in\mc E,|\, x_j=x_i\}|\eeq
This simply says that, for every configuration $\x\in\mc X$, action $\alpha\in\mc A$ and player $i\in\mc V$, it holds
\beq\label{potential} \Phi_{\text c}(\alpha, {\bf x}_{-i})- \Phi_{\text c}({\bf x})=\lambda^{\text c}_i(\alpha, {\bf x}_{-i})-\lambda^{\text c}_i({\bf x})\eeq

The (asyncronous) best response dynamics is a discrete time Markov chain (MC) $X_t$ on the configuration space $\mc X$ where, at every time $t$, a player $i$ is chosen 
uniformly at random and it modifies its action choosing an element uniformly at random within $\mc B^{\text c}_i((X_t)_{-i})$. Denote with $P_{\x,\y}$ the transition matrix (on $\mc X\times\mc X$) of the MC $X_t$ and note that, given $\x\in\mc X$ and $\y=(\alpha, {\bf x}_{-i})$, it holds that
\beq\label{BRtransition} P_{\x,\y}>0\;\Leftrightarrow\;  \Phi_{\text c}(\alpha, {\bf x}_{-i})\geq  \Phi_{\text c}({\bf x})\eeq
From the fact that the potential is not decreasing along the trajectories of $X_t$, it follows the classical result that, with probability $1$, $X_t$ converges in finite time to the set $\mc N$ of Nash equilibria.

The question we pose is: what is the minimal number of agents that if forced to $1$ will ensure almost surely that the best response dynamics reach the state $\1$.

Given a subset $\mc C\subseteq \mc V$, we indicate with $X^{\mc C}_t$ the Markov chain where only the agents in $\mc V\setminus\mc C$ update their action according to the best response rule defined above, while agents in $\mc C$ maintain action $1$. This new MC  takes values in the subset of configurations 
$$\mc X^{(\mc C)}=\{x\in\mc X\,|\, x_i=1\,\forall i\in\mc C\}$$

This restricted game remains potential. This new dynamics will converge too to its set of Nash equilibria $\mc N^{\mc C}=\mc N\cap\mc X^{(\mc C)}$. 

The following definition is the main object of study of this paper. 

\begin{definition}[Sufficient control set]\label{def:sufficient} $\mc C$ is a \emph{sufficient control set} if 
\begin{equation}
\forall x_0\in\mc X^{(\mc C)}, \; 
\P(\exists t\,:\, X^{\mc C}_t=\1,|\, X^{\mc C}_0=x_0)=1
\label{convergence}
\end{equation}
A sufficient control set is \emph{minimal} if none of its strict subsets is a sufficient control set. \\
A sufficient control set is \emph{optimal} if there exists no sufficient control set of strictly smaller cardinality.
\end{definition}

Our objective is to find optimal sufficient control sets. 

To give a more intuitive idea of what control sets resemble, here are a few illustrative examples.

\begin{figure}[h]
\centering
\begin{tikzpicture}

\draw[thick] (-1.3,3) -- (0,4);
\draw[thick] (1.3,3) -- (0,4);
\draw[thick] (-1.8,2) -- (-1.3,3);
\draw[thick] (-1.8,1) -- (-1.8,2);
\draw[thick] (-.8,2) --(-1.3,3);
\draw[thick] (.8,2) --(1.3,3);
\draw[thick] (1.8,2) --(1.3,3);
\draw[thick] (0.3,1) --(.8,2);
\draw[thick] (1.3,1) --(.8,2);

\draw[fill,white] (0,4) circle (.13);
\draw[thick] (0,4) circle (.13);

\draw[fill,red] (-1.3,3) circle (.13);
\draw[thick] (-1.3,3) circle (.13);

\draw[fill, white] (1.3,3) circle (.13);
\draw[thick] (1.3,3) circle (.13);

\draw[fill,white] (-1.8,2) circle (.13);
\draw[thick] (-1.8,2) circle (.13);

\draw[fill,green] (-1.8,1) circle (.13);
\draw[thick] (-1.8,1) circle (.13);

\draw[fill,green] (-.8,2) circle (.13);
\draw[thick] (-.8,2) circle (.13);

\draw[fill,red] (.8,2) circle (.13);
\draw[thick] (.8,2) circle (.13);

\draw[fill,green] (1.8,2) circle (.13);
\draw[thick] (1.8,2) circle (.13);

\draw[fill,green] (0.3,1) circle (.13);
\draw[thick] (0.3,1) circle (.13);

\draw[fill,green] (1.3,1) circle (.13);
\draw[thick] (1.3,1) circle (.13);

\end{tikzpicture}
\caption{Example : on trees, the set of the leaves (in green) is always a valid control set.In red you can see another  valid control set, of size 2, minimal for this particular tree }
\end{figure}
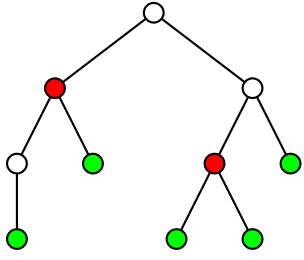

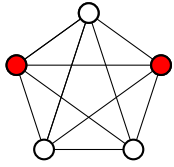
\begin{figure}[h]
\centering
\begin{tikzpicture}
\foreach \i in {90,162,...,450}
{
	\draw (\i+144:1)--(\i:1) -- (\i+72:1);
}
\foreach \i in {90,162,...,450}
{
	\draw[white,fill] (\i:1) circle (.13);
	\draw[thick] (\i:1) circle (.13);
	
}
	\draw[red,fill] (162:1) circle (.13);
	\draw[thick] (162:1) circle (.13);

	\draw[red,fill] (18:1) circle (.13);
	\draw[thick] (18:1) circle (.13);

\end{tikzpicture}
\caption{An example of a clique, the sets of size $\left\lceil \frac{k}{2} \right\rceil-1$ are exactly the minimal control sets}
\end{figure}

\begin{figure}[h]
\centering
\begin{tikzpicture}
\begin{scope}{yshift=.86 cm}
	\draw[thick] (90:1)--(330:1)--(210:1)--(90:1);

\foreach \i in {330,90,210}
{
	\draw[white,fill] (\i:1) circle (.13);
	\draw[thick] (\i:1) circle (.13);
}
\draw[red,fill] (90:1) circle (.13);
\draw[thick] (90:1) circle (.13);
\end{scope}
\begin{scope}[xshift=2.5 cm, yshift=.31cm]

\foreach \i in {90,162,234,306,18}
{
	\draw[thick] (\i:1)--(\i+72:1);
}
\foreach \i in {90,162,234,306,18}
{
	\draw[white,fill] (\i:1) circle (.13);
	\draw[thick] (\i:1) circle (.13);
}
\draw[red,fill] (90:1) circle (.13);
\draw[thick] (90:1) circle (.13);

\end{scope}
\begin{scope}[xshift=5 cm, yshift=-.5 cm]
\draw (0,0)--(0,1)--(1,1)--(1,0)--(2,0)--(2,1);
\foreach \i in {0,1,2}
{
\foreach \j in {0,1}
{
	\draw[white,fill] (\i,\j) circle (.13);
	\draw[thick] (\i,\j) circle (.13);
}
}
\draw[red,fill] (1,1) circle (.13);
\draw[thick] (1,1) circle (.13);

\end{scope}

\end{tikzpicture}
\caption{If all node have degree at most  $2$, then choosing one node per connected component gives a control set}
\end{figure}
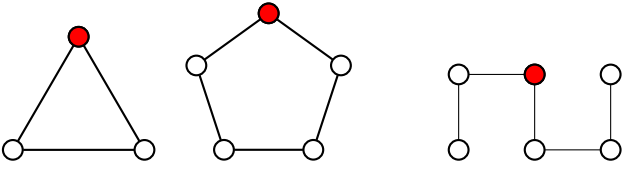

\section{Monotone crusades and valid control sets}\label{sec:monotone}

%

%
%
%

We now define the concept of monotone crusade,  that will play a crucial role in our theory. 

They are particular sequences of states, starting from an initial state, end at $\1$ and such that at every step, the number of their component with value $1$, without ever decreasing the potential. \\
More formally:
\begin{definition}[Adapted Monotone Crusade]\label{def:crusade} Let $\mc C\subseteq \mc V$. 
A \emph{monotone crusade from $\mc C$} is a sequence of vectors ${\bf x}^k\in\mc X$, for $k=0,\dots , m$ such that
\begin{enumerate}
\item ${\mathbf x}^0=\1_{\mc C}$, ${\mathbf x}^m=\1$
\item for every $k=1,\dots ,m-1$ there exists $i_k\in\mc V\setminus\mc C$ such that ${\mathbf x}^{k+1} ={\mathbf x}^k+\delta_{i_k}$ 
\end{enumerate}
Moreover, if given a function $V:\mc X\to \R$, it holds
\begin{enumerate}
\item[3.] $V({\mathbf x}^{k+1})\geq V({\mathbf x}^k)$ for $k=0,\dots , n-1$
\end{enumerate}
then the sequence ${\bf x}^k$ is called a $V$-\emph{adapted} monotone crusade from $\mc C$.
\end{definition}

A few comments on the above definition: 
\begin{remark} All nodes $i_1,\dots ,1_m$ appearing in (2) are necessarily distinct otherwise the condition ${\bf x}^k\in\{0,1\}^n$ for all $k$ would be violated. Indeed, we must have that $\mc V\setminus\mc C=\{i_1,\dots ,1_m\}$ and thus $m=|\mc V\setminus\mc C|$. This allows for a monotone crusade from $\mc C$ to be equivalently characterized by the sequence of nodes $(i_k)$, the induced order on the nodes or by the sequence of increasing support sets $(S_k)$ defined by $S_k=S_{{\bf x}^k}$, for $k=0,\dots , m$ having the property that $S_0=\mc C$ and $S_m=\mc V$. 
\end{remark}
\begin{remark}
We can also define a decreasing version of the monotone crusade where ${\mathbf x}^0=\1$, ${\mathbf x}^m=\1_{\mc C}$ and where. for every $k$, ${\mathbf x}^{k+1} ={\mathbf x}^k-\delta_{i_k}$. This will be called a \emph{decreasing monotone crusade to $\mc C$} ($V$-\emph{adapted} if property 3. holds true). 
\end{remark}

{
\begin{definition}[Valid control set] 
A set $\mathcal{C}$ is $V$-\emph{valid} if there exists a $V$-adapted monotonous crusade from $\mc C$.
\end{definition}

Main goal of the rest of this section is to show that the class of $\Phi_c$-{valid} control sets (we recall that $\Phi_c$ is the potential of the majority game) coincides with the class of sufficient control sets defined in (Definition \ref{def:sufficient}).


The following property is instrumental to our results.}

\begin{lemma}[monotonicity of Coordination Game]\label{lemma:monotonicity}

For all $\x, \y \in\mc X$ and $  i \in \mc{V}$, the following conditions hold:
\begin{enumerate}
\item
if $\mathbf{x}\le \mathbf{y}$ and $\Phi_{\text c}(1,\x_{-i}) \ge \Phi_{\text c}(\mathbf{x})$ then $\Phi_{\text c}(1,\mathbf{y}_{-i}) \ge \Phi_{\text c}(\mathbf{y})$;
\item if $\mathbf{x}\ge \mathbf{y}$ and $\Phi_c(0,\x_{-i}) \ge \Phi_c(\mathbf{x})$ then $\Phi_c(0, \mathbf{y}_{-i}) \ge \Phi_c(\mathbf{y})$.
\end{enumerate}


\end{lemma}
\begin{proof} We only prove the first assertion, the second can be obtained by exchanging the role of $0$ and $1$.
%

If $y_i=1$, we have that $(1, \mathbf{y}_{-i})=\mathbf y$ and there is nothing to prove. If $x_i=1$, then the inequality $ \mathbf{x}\le \mathbf{y}$ ensure that $y_i=1$ and we are thus in the previous case.



We now consider the case when both $x_i$ and $y_i$ have value $0$. 
%
%
%
Note that, if $\mathbf z$ is such that $z_i=0$, the variation of the potential when player $i$ changes its action from $0$ to $1$ can be expressed as 
$$\Phi_c(1,\mathbf{z_{-i}}) - \Phi_c(\mathbf{z})= n_{i,1}(\mathbf{z}) -n_{i,0}(\mathbf{z})$$
where, we recall, $n_{i,1}(\mathbf{z})$ and $n_{i,0}(\mathbf{z})$ are the number of neighbors of $i$ whose action is, respectively, $1$ and $0$.
%
%
%
%
As $\mathbf{x} \le \mathbf{y}$, we have that $n_{i,0}(\mathbf{y}) \le n_{i,0}(\mathbf{x})$ and $n_{i,1}(\mathbf{y}) \ge n_{i,0}(\mathbf{x})$. Hence,
$$\begin{array}{l}\Phi_c(1, \mathbf{y}_{-i}) -\Phi_c(\mathbf{y})=n_{i,1}(\mathbf{y}) -n_{i,0}(\mathbf{y})\\\ge n_{i,1}(\mathbf{x}) -n_{i,0}(\mathbf{x})=\Phi_c(1, \mathbf{x}_{-i}) - \Phi_c(\mathbf{x})\end{array}$$
This yields the thesis.
%
%
%
%
%
%
\qed
\end{proof}

\begin{proposition}[monotonicity for inclusion]\label{prop:superset}
\label{mono}
A superset of a $\Phi_c$-{valid} control set  is a $\Phi_c$-{valid} control set.
\end{proposition}
\begin{proof}
Assume that $\mc C$ is a $\Phi_c$-{valid} control set and let $\mc C'\supseteq \mc C$.
Let ${\bf x}^k$ be a $\Phi_c$-{adapted monotone crusade} from $\mc C$ with associated sequence of points $(i_k)$ for $k=1,\dots , m=n-|\mc C|$ such that
${\bf x}^{k+1}- {\bf x}^k=\delta_{i_k}$ for each $k$.
 Consider the subsequence of points $i_{k_1}, i_{k_2}, \dots , i_{k_{m'}}$ that are in $\mc V\setminus\mc C'$ and put $\y^h=\max\{\1_{\mc C'}, \x^{k_h}\}$. By construction, we have that $\y^{h}\geq \x^{k_{h+1}-1}$ and thus, by Lemma \ref{lemma:monotonicity} and the fact that ${\bf x}^k$ is a  $\Phi_c$-{adapted monotone crusade} from $\mc C$, we have that $\Phi_c({\bf y}^{h})\leq \Phi_c({\bf y}^{h+1})$.

\qed 
\end{proof}
\begin{remark}[full set]
The full set is always a $\Phi_c$-{valid} control set 
\end{remark}

The following clarifies the connection between valid control sets for the majority game and sufficient control sets introduced in the previous section. 

\begin{theorem}\label{theo:valid-sufficient} A subset $\mc C\subseteq \mc V$ is a sufficient control set iff it is a $\Phi_c$-{valid} control set.
\end{theorem}

\begin{proof}
We first show that a sufficient control set is $\Phi_c$-{valid}. If $\mc C$ is a sufficient control set, there exists a sequence of vectors ${\bf y}^0,\dots , {\bf y}^T\in\mc X^{(\mc C)}$ such that ${\bf y}^0=\1_{\mc C}$ and ${\bf y}^T = \1$ that the best response dynamics follows with positive probability. This is equivalent to saying, using  the definition of best response dynamics (see in particular property (\ref{BRtransition})) , that 
\begin{enumerate}
\item ${\bf y}^{k+1}- {\bf y}^k=\pm\delta_{i_k}$ for all $k=0,\dots , T-1$;
\item $\Phi_c({\bf y}^0)\leq\cdots \leq \Phi_c({\bf y}^T)$.
\end{enumerate}
For every $i\not\in\mc C$, define
$$k(i)=\min\{k=1,\dots , T\;|\; {\bf y}^{k}- {\bf y}^{k-1}=\delta_{i}\}$$
that is the first time when agent $i$ change its action to $1$ in the sequence ${\bf y}^t$. Order now the agents in $\mc V\setminus\mc C$ as $i_1,\dots , i_m$ in such a way that 
$k_{i_1}<k_{i_2}<\cdots <k_{i_m}$. Consider the increasing monotone crusade ${\bf x}^h$ associated with the sequence of points $(i_h)$, namely,
$${\bf x}^h=\1_{\mc C}+\sum_{h'\leq h}\delta_{i_h}$$
and notice that ${\bf x}^{h-1}\geq {\bf y}^{k(i_h)-1}$. Since $\Phi_c({\bf y}^{k(i_h)-1})\leq \Phi_c({\bf y}^{k(i_h)})$, it follows from Lemma \ref{lemma:monotonicity} that $\Phi_c({\bf x}^{h-1})\leq \Phi_c({\bf x}^{h})$. This tells us that 
${\bf x}^h$ is a $\Phi_c$-{adapted monotone crusade} from $\mc C$ and, thus, $\mc C$ a $\Phi_c$-valid control set.

We now show that a $\Phi_c$-{valid} control set $\mc C$ is sufficient.\\
 We fix any initial condition ${\bf x}_0\in\mc X^{(\mc C)}$ and we put $\mc C'=S_{{\bf x}_0}$. $\mc C'$ is a superset of $\mc C$ and, on the basis of Proposition \ref{prop:superset}, $\mc C'$ is also a $\Phi_c$-{valid} control set. Let ${\bf x}^k$ be a corresponding $\Phi_c$-{adapted monotone crusade} from $\mc C'$. By the properties of adapted monotone crusades (properties 2. and 3. in Definition \ref{def:crusade}) and the characterization (\ref{BRtransition}) of the transition matrix of the best response dynamics $X_t$, starting from ${\bf x}_0$, the MC $X_t$ will follow such a sequence with positive probability. Thus, from any initial condition, $X_t$ will reach $\1$ with positive probability. A standard result on MC then yields the thesis.

\qed 
\end{proof}

}

\section{Markov chains and backward search algorithms}\label{sec:MC}
The characterization of sufficient control sets through the concept of monotone crusades suggest the possibility that such sets can be found starting from the configuration $\1$, iteratively replacing $1$'s with $0$'s in the attempt to follow backwards a monotone crusade. To this aim we now introduce a family of MC $Z_t^{\epsilon}$ on the binary space $\mc X$, parameterized by $\epsilon\in [0,1]$ that will be the core part of our algorithms.

Transitions of $Z_t^{\epsilon}$ are described as follows:\\ At every discrete time, a node uniformly at random $i$ is activated.
If its neighbors with current action $1$ ($n_{i,1}$) are strictly less than its neighbors with current action $0$ ($n_{i,0}$), it stays still. Otherwise,
 if its  action is $1$ it changes to $0$ with probability $1$, while if its action is $0$, it changes to $1$ with probability $\epsilon$.

The only non zero non trivial transition probabilities of $Z_t^{\epsilon}$ are the following. Given ${\bf x}\in\mc X$, 
\beq\label{zepsilon}
\begin{array}{l}
x_i=1,\, n_{i,1}({\bf x})\geq n_{i,0}({\bf x}) 
\; \Rightarrow\;  P^{\epsilon}_{{\bf x},(0,{\bf x}_{-i})}=1/n 
\\
x_i=0, \, n_{i,1}({\bf x})\geq n_{i,0}({\bf x}) 
\;\;\Rightarrow\;  P^{\epsilon}_{{\bf x},(1,{\bf x}_{-i})}=\epsilon/n
\end{array}
\eeq

In the case when $\epsilon=0$, only transitions from $1$ to $0$ are allowed. In this case, the MC has absorbing points. The relation of these points with sufficient control sets is studied in the next result.

We denote 
$$
\begin{array}{lcl}
\mc Z & =& \{x\in\mc X\;|\; \P(\exists t_0\,:\, Z^0_{t_0}=x\;|\; Z_0^0=\1)>0\}\\[1pt]
\mc Z^{\infty} & = & \{x\in\mc X\;|\; \P(\exists t_0\,:\, Z^0_{t}=x \,\forall t\geq t_0\;|\; Z^0_0=\1)>0\}
\end{array}
$$
the sets of reachable and absorbing state of the chain $Z^{0}$.

\begin{theorem}\label{theo:Z} The following facts hold:
\begin{enumerate}
\item $\mc C$ is a sufficient control set iff $\1_{\mc C}\in\mc Z$;
\item $\mc C$ is a minimal sufficient control set iff $\1_{\mc C}\in\mc Z^{\infty}$;
\end{enumerate}
\end{theorem}
\begin{proof}

(1):
By definition, if $\x=\1_{\mc C}\in \mc Z $, there exists a sequence of configuration vectors $y^k$, for $k=0,\dots , m$ such that 
${\bf y}^0=\1$ and ${\bf y}^m = \1_{\mc C}$ satisfying the properties
\begin{enumerate}
\item ${\bf y}^{k}- {\bf y}^{k+1}=\delta_{i_k}$ for all $k=0,\dots , m-1$;
\item $\Phi_c({\bf y}^0)\geq\cdots \geq \Phi_c({\bf y}^T)$.
\end{enumerate}
Then $x^k=y^{m-k}$ is a $\Phi^c$-adapted monotone crusade from 
$\mc C$ and this yields that $\mc C$ is a $\Phi_c$-valid control set and thus also a sufficient control set by virtue of Theorem \ref{theo:valid-sufficient}. Inverting this argument we prove the other implication. 

%
%

(2): If $\mc C$ is a minimal sufficient control set, we know from (1) that $\1_{\mc C}\in\mc Z$. If, by contradiction, $\1_{\mc C}\not\in\mc Z^{\infty}$, then, from $\x=\1_{\mc C}$, the MC $Z^0_t$ could reach, in one step, a different state $\x'=\1_{\mc C'}$ with $\mc C'\subsetneq \mc C$. This contradicts minimality. 

\end{proof}

Theorem \ref{theo:Z} allows to reformulate the optimization problem as follows:
\beq\label{optimization}\min\limits_{\x\in\mc Z}||\x||_1\eeq
where $||x||_1=\sum_ix_i$. Optimal sufficient control sets $\mc C$ are those for which $\1_{\mc C}$ solves (\ref{optimization}).

The MC $Z^0_{t}$ is naturally related to the minority game whose definition we briefly recall thus: Given an undirected graph $\mc G=(\mc V, \mc E)$, we define the minority game on $\mc G$ as the binary game where each player $i\in\mc V$ has utility $\lambda^{\text a}_i:\mc X\to \R$ given by
$$\lambda^{\text a}_i({\bf x})=|\{j\in N_i\,|\, x_j\neq x_i\}|$$
that is simply the number of neighbors with which $i$ is in disagreement. 
We denote by $\mc N_{\text a}$ the set of Nash equilibria of the minority game.
This game is potential with a potential $\Phi_{\text a}$ that is just the opposite than the potential of the majority game:
\beq \Phi_{\text a}({\bf x})=-\Phi_{\text c}(\x)\eeq

The following property clarifies the relation of the minority game with our problem.
\begin{proposition}\label{prop:minority} 
$\mc N_{\text a}\subseteq \mc Z$: Nash equilibria of the minority game are valid control sets.
\end{proposition}
\begin{proof}
Let $\x\in\mc N_a$ and let ${\bf y}^k$, for $k=1,\dots , m$ be any decreasing monotone crusade from $\1$ to ${\bf x}$. By construction, ${\bf x}\le {\bf y}^k \le \1$ for all $k$. 
Consider the sequence of nodes $(i_k)$ such that ${\bf y}^{k-1}-{\bf y}^{k}=\delta_{i_k}$. We have that ${\bf  x}_{i_k}=0$ for every $k$.
 For all $k$, since ${\bf x}$ is a Nash equilibrium then $0$  is in $i_k$'s minority best response, thus $\Phi_a(\x) \ge \Phi_a(1,\x_{-i})$, or equivalently, $\Phi_c(\x) \le \Phi_c(1,\x_{-i})$. By Lemma \ref{lemma:monotonicity}, it follows that $\Phi_c(\y^{k}) \le \Phi_c(1,\y^{k}_{-i}) = \Phi_c(\y^{k-1})$. Therefore $\x^h=\y^{m-h}$ is a 
$\Phi_c$-adapted monotone crusade from $\mc C=S_{\x}$. By virtue of Theorem \ref{theo:Z}, we have that $\x\in\mc Z$.

\end{proof}

We have the following simple but not obvious consequence.
\begin{corollary}[existence]
For any graph, there exist a sufficient control set whose size is less or equal to half the number of nodes. 
\end{corollary}
\begin{proof}
The minority game, being potential, admits at least one Nash equilibrium $\x\in\mc N_a$. By symmetry  $\tilde{\x}=\1-\x$ is also a Nash equilibrium. Proposition \ref{prop:minority} tells us that they are both sufficient control sets and one of the two has the required size property. 
\end{proof}

Theorem \ref{theo:Z} and Proposition \ref{prop:minority} can not be improved further as shown by the following examples.

In the following  figures  \ref{fig:ex1} and \ref{fig:ex1b} we report two examples of optimal sufficient control sets that are not Nash equilibria of the minority game, while figure \ref{fig:ex2} is a counter example to the assumption that $\mc Z^\infty$ elements are minimal.

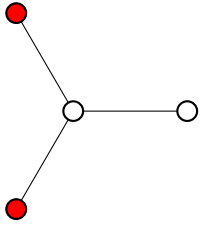
\begin{figure}[h]

\centering
\begin{tikzpicture}

\foreach \i in {0,120,240}
{
\draw (\i:0) -- (\i:1.5);
\draw[fill] (\i:1.5) circle(.13);
}
\draw[fill,white] (0,0) circle (.13);
\draw[fill,white] (0:1.5) circle(.13);
\draw[thick] (0,0) circle (.13);
\draw[thick] (0:1.5) circle(.13);

\draw[fill,red] (120:1.5) circle(.13);
\draw[fill,red] (240:1.5) circle(.13);
\draw[thick] (120:1.5) circle(.13);
\draw[thick] (240:1.5) circle(.13);

\end{tikzpicture}
\caption{This set is a minimal control set for the inclusion, yet the rightmost node is not in its best response, making it not a Nash }
\label{fig:ex1}
\end{figure}
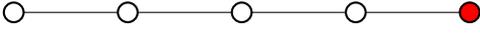
\begin{figure}[h]
\centering

\begin{tikzpicture}
\draw (0,0) -- (6,0);

\foreach \i in {0,1.5,3,4.5}
{

\draw[fill,white] (\i,0) circle(.13);
\draw[thick] (\i,0) circle(.13);
}
\draw[fill, red] (6,0) circle(.13);
\draw[thick] (6,0) circle(.13);
\end{tikzpicture}\\

\caption{This set is optimal, but not Nash}
\label{fig:ex1b}
\end{figure}

%
%
%


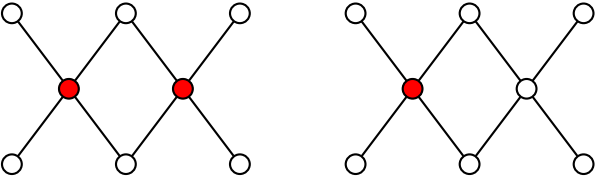
\begin{figure}[h]

\centering
\begin{tikzpicture}{scale=.8}
\draw[thick] (-.75,1) -- (0,0)--(.75,1)--(1.5,0)--(.75,-1)--(0,0)--(-.75,-1);
\draw[thick](2.25,1)--(1.5,0)--(2.25,-1);

\draw[fill,red] (0,0) circle (.13);
\draw[thick] (0,0) circle (.13);

\draw[fill,white] (.75,1) circle (.13);
\draw[thick] (.75,1) circle (.13);

\draw[fill,white] (.75,-1) circle (.13);
\draw[thick] (.75,-1) circle (.13);

\draw[fill,red] (1.5,0) circle (.13);
\draw[thick] (1.5,0) circle (.13);

\draw[fill,white] (2.25,1) circle (.13);
\draw[thick] (2.25,1) circle (.13);

\draw[fill,white] (2.25,-1) circle (.13);
\draw[thick] (2.25,-1) circle (.13);

\draw[fill,white] (-.75,1) circle (.13);
\draw[thick] (-.75,1) circle (.13);

\draw[fill,white] (-.75,-1) circle (.13);
\draw[thick] (-.75,-1) circle (.13);

\end{tikzpicture}
\hspace{1cm}
\begin{tikzpicture}{scale=.8}

\draw[thick] (-.75,1) -- (0,0)--(.75,1)--(1.5,0)--(.75,-1)--(0,0)--(-.75,-1);
\draw[thick](2.25,1)--(1.5,0)--(2.25,-1);

\draw[fill,red] (0,0) circle (.13);
\draw[thick] (0,0) circle (.13);

\draw[fill,white] (.75,1) circle (.13);
\draw[thick] (.75,1) circle (.13);

\draw[fill,white] (.75,-1) circle (.13);
\draw[thick] (.75,-1) circle (.13);

\draw[fill,white] (1.5,0) circle (.13);
\draw[thick] (1.5,0) circle (.13);

\draw[fill,white] (2.25,1) circle (.13);
\draw[thick] (2.25,1) circle (.13);

\draw[fill,white] (2.25,-1) circle (.13);
\draw[thick] (2.25,-1) circle (.13);

\draw[fill,white] (-.75,1) circle (.13);
\draw[thick] (-.75,1) circle (.13);

\draw[fill,white] (-.75,-1) circle (.13);
\draw[thick] (-.75,-1) circle (.13);

\end{tikzpicture}
\caption{Two sufficient control sets on the same graph, both in $\mc Z^\infty$ only the one on the right is minimal.}
\label{fig:ex2}
\end{figure}


The above considerations imply that we can not directly use the dynamics $Z^0_t$ as our algorithm to find optimal sufficient control sets as its absorbing state $\mc Z^{\infty}$ are not even minimal in general. To overcome this difficulty we will instead use the MC $Z^{\epsilon}_t$ with $\epsilon >0$. As we will see, the presence of the transition in the opposite direction will permit the algorithm to make a full search on the set $\mc Z$ and not to remain stacked in non optimal configurations. This is guaranteed by the following final result.

%
%
%
%
%

\begin{theorem} Let $\epsilon >0$. The following facts hold:
\begin{enumerate}
\item $Z^{\epsilon}_t$ is ergodic inside the set of states $\mc Z$;
\item $Z^{\epsilon}_t$ is time-reversible and its unique invariant probability is given by $\mu^\epsilon ({\bf x}) := K\epsilon^{ ||\x||_1}$ where $K>0$ is the normalization constant;
\item For $\epsilon\to 0+$, $\mu^{\epsilon}$ converges in law to a probability measure $\mu$ concentrated on the subset $\argmin_{{\bf x}\in\mc Z}||\x||_1$.
\end{enumerate}
\end{theorem}
\begin{proof} 
(1): We first show that $\mc Z$ is invariant by $Z^{\epsilon}_t$. First notice that the transitions that decrease the number of $0$'s in  $Z^{\epsilon}$ are also possible in $Z^0$. We thus only have to check the invariance for the transitions that increase the number of $1$'s in $Z^{\epsilon}$. Fix $\bf x\in \mc Z$, suppose that $x_i=0$ and put $\bf y= (1,{\bf x}_{-i})$. If $P^{\epsilon}_{{\bf x},\bf y}>0$, it follows from 
(\ref{zepsilon}) that $n_{i,1}(\bf x)\geq n_{i,0}(\bf x)$. Since $\bf x\in \mc Z$, there exists a sequence of configuration vectors $\bf x^k$, for $k=0,\dots , m$ such that 
${\bf x}^0=\1$ and ${\bf x}^m = \1_{\mc C}$ satisfying the properties
\begin{enumerate}
\item ${\bf x}^{k}- {\bf x}^{k+1}=\delta_{i_k}$ for all $k=0,\dots , m-1$;
\item $\Phi_c({\bf x}^0)\geq\cdots \geq \Phi_c({\bf x}^T)$.
\end{enumerate}
There must exist $\bar k$ such that $i_{\bar k}=i$. Consider now $\bf y^h$ defined by
$$\bf y^h=\left\{\begin{array}{ll} \bf x^h\quad &{\rm for}\, h\leq \bar k\\ \bf x^{h+1}+\delta_{i} \quad &{\rm for}\, h> \bar k\end{array}\right.$$
Notice that ${\bf y}^{T-1}=\bf y$
Using again Lemma \ref{lemma:monotonicity}, we have that $\Phi_c({\bf y}^0)\geq\cdots \geq \Phi_c({\bf y}^{T-1})$ so that $\bf y\in\mc Z$.

Ergodicity follows now from the observation that $Z^\epsilon$ verify the following property: for all transition with non zero probability, also the reversed transition has non zero probability. If $\bf x,\bf y\in\mc Z$, by definition of $\mc Z$, they are both reachable by $Z^0$ (and thus also by $Z^{\epsilon}$) starting from $\1$. Therefore, with positive probability, under the MC $Z^{\epsilon}$ it is possible to move from $\bf x$ to $\1$ and then to $\bf y$. 

(2) We need to show that for all $\bf x,\, \bf y\in\mc Z$, it holds $$ \epsilon^{ ||\x||_1}P^{\epsilon}_{{\bf x},\bf y}=\epsilon^{ ||\y||_1}P^{\epsilon}_{{\bf y},\bf x}$$
Since the only non-trivial transitions of $Z^{\epsilon}$ are those in (\ref{zepsilon}), we can reduce to check it in the case when, for some $i\in\mc V$, $x_i=1$ and $\bf y=(0, \bf x_{-i})$. In this case we obtain
that $$ \epsilon^{ ||\x||_1}P^{\epsilon}_{{\bf x},\bf y}= \frac{\epsilon^{ ||\x\|_1}}{n},\quad \epsilon^{ ||\y||_1}P^{\epsilon}_{{\bf y},\bf x}= \epsilon^{ ||\x||_1-1}\frac{\epsilon}{n}=\frac{\epsilon^{ ||\x||_1}}{n}$$
This proves the claim.

(3) follows from the expression of $\mu^{\epsilon}$.

%
%
%
%
%
%
%
\qed
\end{proof}

By virtue of the reformulation (\ref{optimization}), we have that, for small $\epsilon$ and sufficiently large $t$, the MC $Z^{\epsilon}_t$ will be 'most of the time' in configurations whose support are optimal sufficient control sets. Thsi observation is at the base of a practical algorithm described in the next section.

%
%
%
%
%
%
%

\section{Simulations}\label{sec:simulations}

We have implemented an iterative algorithm based on the MC $Z^{\epsilon}_t$. For the sake of increasing the speed to convergence, we actually considered a modification of $Z^{\epsilon}_t$ with all trivial self-loop transitions removed. This induces a little modification in the invariant probability, but do not affect the minimal set that is the return of the algorithm. Details will be reported elsewhere. The algorithm keeps track of the best configuration (the smallest $||x||_1$) found so far.

We have applied the algorithm to random realizations of Erd\"os-R\'eniy graphs with different number of nodes $n$ and probability $p=1/2$. For every value of $n$, we ran $500$ executions on $20$ randomly generated graphs. The algorithm is stopped after $100n$ iterations, using for epsilon the constant value {$\epsilon = 0.2 $.} 

%

As a point of comparison, we computed a benchmark optimum consisting of the exhaustive optimum for small graphs, and a much longer execution on bigger graphs. Figure \ref{fig} shows the average values of the size of the sufficient control sets computed by the algorithm. We compare it with the benchmark optimum and also with the result obtained looking at the very last step of the algorithm. This plot shows a remarkable performance of the algorithm that in linear time gets quite close to the optimum. 
It also shows that the MC, though fluctuating, as $\epsilon >0$, still remains close to the optimal configurations. The important question of how to tune the parameter $\epsilon$ for optimize performance has not been addressed here.

Finally notice how optimal sufficient control sets are scaling linearly with respect to the size of the graph. This suggests that Erd\"os-R\'eniy graphs are somewhat 'difficult' to control in this sense, in other terms they show resilience to this type of external actions.

\begin{figure}
\includegraphics[scale=.7]{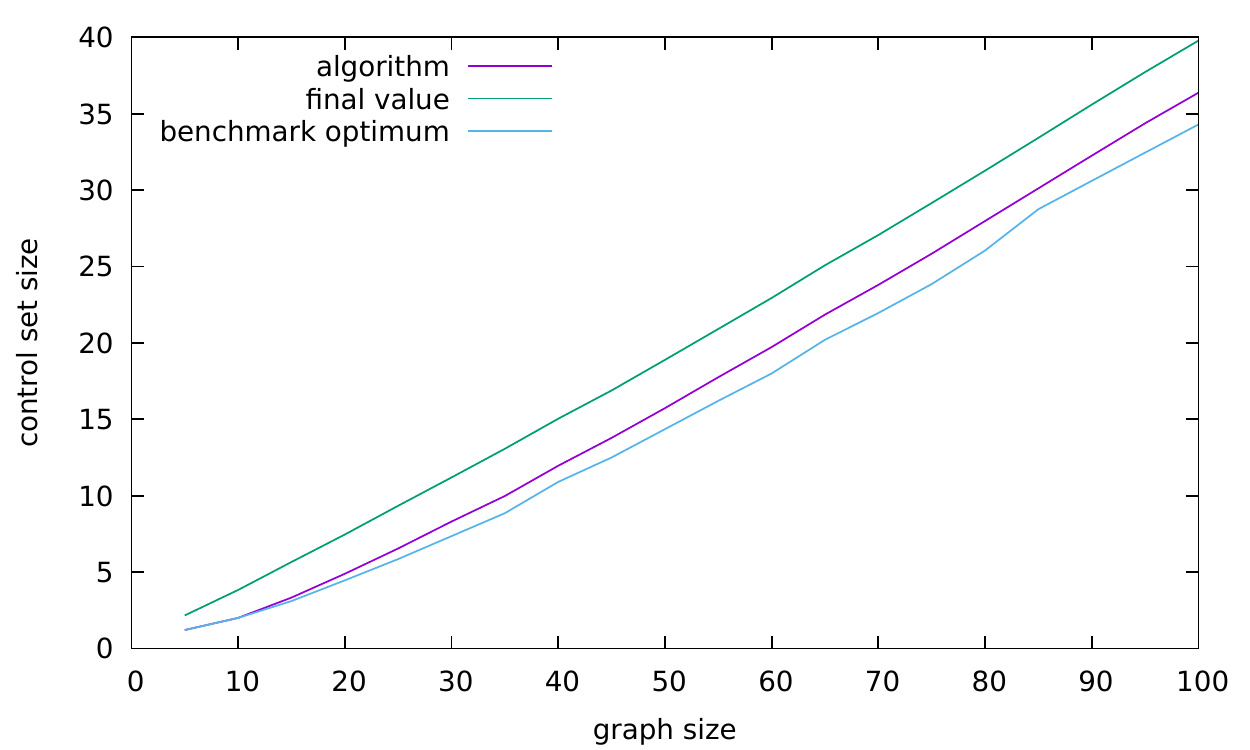}
\caption{In blue, the benchmark optimum. In purple, the minimal control set encountered in $100n$ steps. In green, the size of the control set at the end. 
}
\label{fig}
\end{figure}

%
%
%

\section{Conclusion}\label{sec:conclusions}

We have formulated the problem of finding, in a network coordination game, the minimum number of players to be controlled in order to drive the system from one Nash equilibrium to another one. To the scope, we have designed a low complexity randomized algorithm and proven its convergence properties. We have finally carried on some numerical simulations corroborating the results.

Many challenging issues naturally pop up from our analysis and simulations. Erd\"os-R\'eniy graphs have exhibited optimal sufficient control sets growing linearly in the size of the graph. It would be of interest if this could be proven analytically, as well if this could be extended to other family of graphs, connecting such resilient phenomena to topological properties. 

The problem studied in this paper is an instance of a more general problem of studying the effect of control actions in evolutionary game theory. Future research will be in this direction analyzing similar control problems for general games with strategic complementarity as well strategic substitutes. 

%
%
%
%
%

\bibliography{bib}

\end{document}